\documentclass[11pt]{article}

\usepackage[margin=1.1in]{geometry}
\usepackage{amsmath,amssymb,amsthm}
\usepackage{mathtools}
\usepackage{booktabs}
\usepackage[numbers]{natbib}
\usepackage[hidelinks]{hyperref}
\usepackage{microtype}
\usepackage{pgf}
\usepackage{type1cm}

\theoremstyle{plain}
\newtheorem{theorem}{Theorem}
\newtheorem{proposition}[theorem]{Proposition}
\newtheorem{corollary}[theorem]{Corollary}
\theoremstyle{definition}
\newtheorem{definition}[theorem]{Definition}
\theoremstyle{remark}
\newtheorem{remark}[theorem]{Remark}

\newcommand{\R}{\mathbb{R}}
\newcommand{\E}{\mathbb{E}}
\newcommand{\Cov}{\operatorname{Cov}}
\newcommand{\Var}{\operatorname{Var}}
\newcommand{\Sig}{\mathbb{S}}
\newcommand{\shuf}{\mathbin{\text{\rlap{$\sqcup$}$\sqcup$}}}
\newcommand{\ip}[2]{\langle #1,\, #2\rangle}
\newcommand{\CE}{\mathrm{CE}}

\title{\bfseries Path Portfolio Optimization:\\
Defect, Lift, and the Price of Path Complexity}

\author{Miquel Noguer i Alonso\\
\normalsize Artificial Intelligence Finance Institute}

\date{\today}

\begin{document}
\maketitle

\begin{abstract}
This paper builds \emph{path portfolio optimization}: portfolio theory on a
path-first framework in which the signature is the universal coordinate of the
price path, and asks whether it survives estimation. A portfolio is a linear
functional of the signature, so the control lives in a truncated tensor
algebra, the covariance of signature coordinates is the non-group-like part of
the expected signature --- a defect form --- and the whole mean--variance
problem becomes a linear system in one tensor. Two structural results follow.
The lift is the execution convention: the gap between the Marcus and forward
lifts, contracted with portfolio weights, is Fernholz's excess growth rate
exactly, so excess growth is the geometricity defect of the portfolio map. And
the antisymmetric block at level two is lift-invariant pathwise, so
directional signals are convention-free while variance signals and ruin are
not. The empirical finding is a dimensional trade-off. With the expected
signature known, quadratic path functionals raise the certainty equivalent
elevenfold for a pair of assets and sixtyfold for a cross section of twenty;
with it estimated, the unregularized policy is severely negative until the
sample exceeds roughly six observations per parameter, and shrinkage flips
from harmful in the pair to indispensable in the cross section. The entire
gain sits in the symmetric block, which is convexity in the terminal increment
rather than path-dependence; the path-dependent antisymmetric block earns
nothing when the driver has no expected area. And the sample-size floor
belongs to unstructured estimation rather than to path complexity: an
estimator that fits only the generator of the driver and rebuilds the expected
signature recovers almost all of the attainable value at barely one
observation per parameter.
\end{abstract}

\vspace{1em}
\noindent\textbf{Keywords:} path signatures; rough paths; shuffle algebra;
expected signature; portfolio optimization; stochastic portfolio theory;
excess growth rate; Marcus and It\^o lifts; Hawkes processes; estimation
risk; shrinkage; sample complexity.

\newpage

\section{Introduction}

Classical portfolio theory takes a weight vector and a covariance matrix as
primitives \citep{markowitz1952}. Everything path-dependent --- momentum,
drawdown control, trend, variance trading, lead--lag --- arrives afterwards as
a signal bolted onto a static optimizer. The alternative is to make the path
the primitive: the signature is a universal coordinate system for continuous
functionals of a path \citep{chen1957,lyonsqian2002,levin2013}, and in those
coordinates path-dependent and static allocation stop being different problems.

The point of departure is \citet{noguer2026paths}, which argues that
geometricity is a first-order algebraic property whose obstructions are second
order: a bracket when the lift is non-geometric, and a covariance when one
averages over random paths. This paper asks what each obstruction does to a
portfolio. Each turns out to be an object portfolio theory already prices under
another name --- and then, because being structurally elegant is not the same
as being usable, it asks how much of the resulting advantage survives having
to estimate anything.

\paragraph{Spine.} \emph{The covariance obstruction is risk; the bracket
obstruction is execution; and the dimension of the strategy space is the price
of admitting either.} We call the resulting framework \emph{path portfolio
optimization}, by analogy with stochastic portfolio theory: where SPT takes
the market weights as its primitive, path portfolio optimization takes the
signature.

\paragraph{Contributions.}
\begin{enumerate}
\item \textbf{Excess growth is the geometricity defect}
(Corollary~\ref{cor:excess}). The Marcus--forward lift gap at level two is one
half the quadratic covariation; contracted with portfolio weights it is
exactly Fernholz's $\gamma^{*}$. Stochastic portfolio theory's central quantity
is the bracket obstruction of \citet{noguer2026paths}, and a portfolio theory
written on a geometric lift has already spent it.

\item \textbf{Direction is convention-free, risk and ruin are not}
(Corollaries~\ref{cor:area},~\ref{cor:ruin}). The antisymmetric level-two block
is identical under both lifts \emph{pathwise}, to machine precision, so no
lead--lag estimate depends on the It\^o/Marcus choice. The symmetric block does,
and so does solvency: Marcus-lifted wealth cannot be extinguished by a jump,
forward-lifted wealth can. Whether a model admits this representation of gap risk is determined by the
chosen lift --- that is, by the assumed execution convention --- rather than by
the dynamics alone.

\item \textbf{The price of path complexity is dimensional, and it is
paid in the symmetric block} (Section~\ref{sec:cost}). With $\E[\Sig]$ known,
level-two controls raise the certainty equivalent by $11.15\times$ at $d=2$
and $59.6\times$ at $d=20$; a word-class decomposition
(Proposition~\ref{prop:wordclass}) puts all of that gain in the symmetric
words and exactly none of it in the antisymmetric ones. With $\E[\Sig]$
estimated, the raw plug-in is negative below roughly six observations per
parameter --- median $-306\times$ the oracle at $M/p=1.19$, with an
interquartile range that does not come near zero --- while ridge shrinkage
that destroys half the value at $d=2$ delivers $0.58\times$ the oracle there.
A model-consistent estimator that fits only the generator and rebuilds
$\E[\Sig]$ by exponential and shuffle recovers $0.95\times$ at the same
sample size (Proposition~\ref{prop:model}): the floor is a property of
unstructured estimation of the risk form, not of path complexity.

\item \textbf{The optimizer exists exactly when the class is coherently
priced} (Section~\ref{sec:hilbert}). On the infinite-dimensional closure of
the strategy space, the mean--variance value is finite if and only if
$\mu\in\operatorname{Range}(\Gamma^{1/2})$ and attained if and only if
$\mu\in\operatorname{Range}(\Gamma)$; the maximal Sharpe ratio of the
class is $\|\Gamma^{\dagger/2}\mu\|$, and failure of the range condition
is an asymptotic arbitrage inside the model. In log coordinates the leading
obstruction to group-likeness of $\E[\Sig]$ is exactly half the covariance
of terminal increments (Proposition~\ref{prop:logdefect}).

\item \textbf{Self-exciting allocation closes, with a stated failure mode}
(Section~\ref{sec:hawkes}). The Hawkes closure of \citet{noguer2026paths} makes
myopic path-dependent allocation a linear solve in an augmented state. What
fails near criticality is not the closure but the habit of substituting its
stationary limit: at branching ratio $0.9$ the stationary variance overstates
the realized one by a factor of $4.9$ on a twenty-unit horizon.

\item \textbf{The cross-area is a signed, capped directional portfolio ---
but only against sign excitation} (Section~\ref{sec:area}). It detects the sign
of excitation asymmetry and is null under symmetric excitation; being
reversal-odd, its edge is bounded by entropy production under the
coarse-graining hypotheses of Proposition~\ref{prop:ceiling}, which bind the
one-step increments the cross-area accumulates rather than the path
functional directly (Remark~\ref{rem:ceilscope}). The second negative result is that excitation
which
raises intensity without carrying direction moves the cross-area of the
counting path ($t=15.99$) and not that of the price path ($t=0.96$).
Reflexivity in activity and reflexivity in direction are different resources,
and only the second is tradable.
\end{enumerate}

\paragraph{Related work.} Signature-linear trading strategies with an
exact shuffle-based mean--variance solution are due to \citet{futter2023sig},
and linear path-functional portfolios in stochastic portfolio theory, with
universality and convex tractability, to \citet{cuchiero2023moller}; a kernel
generalization is \citet{mucacirone2025}. Relative to these, the present paper
contributes no new optimizer. What it adds is a portfolio interpretation of the risk form through the
geometricity defect and a connection between excess growth and the lift gap
(Sections~\ref{sec:defect}--\ref{sec:lift}), the dimensional estimation study
(Section~\ref{sec:cost}), the corrected use of Hawkes closures
(Section~\ref{sec:hawkes}), and the activity-versus-direction decomposition of
the cross-area trade (Section~\ref{sec:area}), whose statistic itself goes
back to \citet{gyurko2013}. The dimensional study stands in an explicit
relation to the virtue-of-complexity result of \citet{kelly2024virtue}, spelled
out in Remark~\ref{rem:voc}.

\paragraph{What this paper is not.} No market data is used anywhere. All
parameters are invented and internally consistent. The estimation results are
about sampling error only: transaction costs, non-stationarity and model
misspecification would all make the picture worse, never better.

\section{Signature-linear portfolios}

Let $X:[0,T]\to\R^{d}$ be a log-price path of finite $p$-variation, possibly
with jumps, $\Sig(X)_{0,T}$ its signature and $\Sig^{\le m}$ the truncation at
level $m$. Words in $\{1,\dots,d\}$ index coordinates, $\ip{w}{\Sig}$ is the
coefficient of $w$, and $\shuf$ is the shuffle product.

\begin{definition}[Signature-linear control]
A \emph{signature-linear portfolio} of order $m$ is
$G_{\ell}(X) = \ip{\ell}{\Sig^{\le m}(X)_{0,T}}$.
\end{definition}

The class is asymptotically unrestrictive: signature-linear functionals are
dense in continuous functionals on compacta of unparameterized paths, the
standard universality argument behind signature methods in finance
\citep{levin2013,kalsi2020,lyons2019nonparametric,cuchiero2023signature}. Level
one reproduces buy-and-hold. Level two reproduces
$\int_{0}^{T}X^{i}\,dX^{j}$ --- momentum when $i=j$, lead--lag when $i\neq j$.
Everything a practitioner calls ``a signal times a position'' lives at level
two or above. Throughout, $p = d + d^{2}$ denotes the number of level-$\le 2$
words, the parameter count of the problem.

\section{The defect form}
\label{sec:defect}

\begin{definition}[Defect form]
$D(u,v) = \ip{u \shuf v}{\E[\Sig(X)]} - \ip{u}{\E[\Sig(X)]}\ip{v}{\E[\Sig(X)]}$.
\end{definition}

\begin{proposition}[Defect form is covariance; \textsc{Proved, known}]
\label{prop:defect}
For a geometric lift with $\E\|\Sig^{\le 2m}\|<\infty$ and words $u,v$ of
length at most $m$, $D(u,v) = \Cov(\ip{u}{\Sig},\ip{v}{\Sig})$; hence
$\Cov(G_{\ell},G_{\ell'}) = \ell^{\top}D\ell'$.
\end{proposition}

\begin{proof}
The shuffle identity $\ip{u}{\Sig}\ip{v}{\Sig}=\ip{u\shuf v}{\Sig}$ holds
pathwise; take expectations and subtract the product of means. Words in
$u\shuf v$ have length at most $2m$.
\end{proof}

\begin{remark}[Global Hopf formulation]
\label{rem:hopf}
The coordinate statement has a coordinate-free form. In the shuffle Hopf
algebra the shuffle product is adjoint to the deconcatenation coproduct
$\Delta$ \citep{reutenauer1993}, so with $\bar\Sig=\E[\Sig(X)]$ the whole
defect form is the single tensor
\[
\mathbf D \;=\; \Delta\bar\Sig \;-\; \bar\Sig\otimes\bar\Sig
\;\in\; T((\R^{d}))\otimes T((\R^{d})),
\qquad
D(u,v)=\ip{u\otimes v}{\mathbf D}.
\]
Group-likeness \emph{is} the equation $\Delta g=g\otimes g$, so
$\mathbf D=0$ iff $\bar\Sig$ is group-like --- the degeneracy half of
Corollary~\ref{cor:psd} becomes definitional. And since a geometric signature
is group-like pathwise, $\Delta\Sig=\Sig\otimes\Sig$, while $\Delta$
commutes with expectation,
\[
\mathbf D=\E[\Delta\Sig]-\E[\Sig]\otimes\E[\Sig]
=\E[\Sig\otimes\Sig]-\E[\Sig]\otimes\E[\Sig],
\]
the covariance tensor of the signature: Proposition~\ref{prop:defect} in two
lines, with no words. The coordinate proof is retained above because it shows
where the truncation level $2m$ enters.
\end{remark}

\begin{remark}[Attribution]
This is not new, and two lines of work make it precise.
\citet{futter2023sig} represent a trading strategy as a linear functional of
the signature and solve the resulting mean--variance problem in closed form,
with the variance computed exactly by shuffling against an expected signature;
Equation~\eqref{eq:mv} below is the static special case of their Theorem~3.1.
\citet{cuchiero2023moller} introduce linear path-functional portfolios in
stochastic portfolio theory, prove their universality, and observe that
mean--variance and log-wealth optimization over the class reduce to convex
quadratic programs. Shuffle-based moment computations against expected
signatures are likewise standard in signature pricing and execution
\citep{kalsi2020,lyons2019nonparametric,cuchiero2023signature}, and the
group-likeness characterization in Corollary~\ref{cor:psd} is the
expected-signature theory of \citet{lyonsni2015} and \citet{chevyrev2016}.
The proposition is stated here as shared infrastructure, not a contribution.
What this paper adds downstream of it is the reading of $D$ as the
\emph{geometricity defect} of \citet{noguer2026paths} --- so that risk is
literally the failure of the mean path to be a path --- and the estimation
study of Section~\ref{sec:cost}, which neither of the works above carries out.
\end{remark}

\begin{corollary}[\textsc{Proved}]
\label{cor:psd}
$D$ is symmetric positive semidefinite, and $D\equiv 0$ if and only if
$\E[\Sig(X)]$ is group-like, which holds if and only if the reduced path is
deterministic.
\end{corollary}

So an expected signature is a legitimate path exactly when there is no risk to
price. The useful consequence is that the mean and the risk of the entire
strategy space are two linear readings of one tensor: with
$\mu_{w}=\ip{w}{\E[\Sig]}$ and risk aversion $\gamma$,
\begin{equation}
\ell^{*} = \frac{1}{\gamma}D^{-1}\mu,
\qquad
\CE = \frac{1}{2\gamma}\,\mu^{\top}D^{-1}\mu .
\label{eq:mv}
\end{equation}

\subsection{Calibration and check}
\label{sec:check}

Take $d=2$, $T=1$, $X_{t}=bt+\sigma W_{t}$ with $b=(0.06,0.02)^{\top}$ and
$\Sigma=\begin{psmallmatrix}0.040&0.012\\0.012&0.090\end{psmallmatrix}$, so
that $\E[\Sig_{T}]=\exp(T(b+\tfrac12\Sigma))$ in closed form. Against
$200{,}000$ simulated paths the closed-form coordinates
$\E[\Sig^{1}]=0.060000$, $\E[\Sig^{12}]=0.006600$, $\E[\Sig^{1122}]=0.000480$
return $0.060126$, $0.006738$, $0.000475$. The defect form computed from
$\E[\Sig]$ alone reproduces the sample covariance of the six level-$\le2$
payoffs to a maximum relative deviation of $0.0036$; its eigenvalues are all
strictly positive with condition number $134.3$; on a deterministic path it
returns $2.78\times10^{-17}$, confirming Corollary~\ref{cor:psd}. The pathwise
shuffle identity holds to $2.0\times10^{-15}$ and Chen concatenation to
$6.7\times10^{-16}$.

At $\gamma=3$ the optimizer \eqref{eq:mv} gives
$\ell^{*}=(0.0448,\,0.0007,\,7.9375,\,-1.1689,\,-1.1689,\,3.8578)$ on the basis
$((1),(2),(11),(12),(21),(22))$, with predicted mean and variance $0.334685$
and $0.111562$ against realized $0.333902$ and $0.111666$. The weights on
$(12)$ and $(21)$ are numerically identical, so the antisymmetric --- Lévy
area --- component of the optimal control is \emph{exactly zero}: under a
driver whose area has zero mean the optimizer declines to trade direction and
holds only the symmetric, variance-like combination. Section~\ref{sec:area} is
that observation with the mean put back in. The largest weight sits on $(11)$.
Under a geometric lift the shuffle relation makes that word exactly
$\tfrac12(X^{1}_{T}-X^{1}_{0})^{2}$, the squared terminal log return, whose
mean is dominated by the variance: the gain comes from convexity, not
forecasting. It is not the variance swap $[X^{1},X^{1}]_{T}$, which by
Proposition~\ref{prop:complement} enters the investable universe only under
the forward lift --- a distinction Section~\ref{sec:cost} shows to be the
whole content of the level-two premium.

\section{The Hilbert geometry of the strategy space}
\label{sec:hilbert}

The finite truncation of Section~\ref{sec:defect} sits inside an
infinite-dimensional problem, and making that explicit settles a question the
finite statement quietly assumes: when does the optimizer \eqref{eq:mv} exist
at all? Fix summable weights $\omega_{w}>0$ and let $\mathcal H_{\omega}$ be
the Hilbert space of coefficient sequences $\ell=(\ell_{w})_{w}$ with
$\|\ell\|^{2}=\sum_{w}\ell_{w}^{2}\omega_{w}<\infty$, so that
$\Phi(X)=(\ip{w}{\Sig(X)})_{w}$ is the signature feature map and
$G_{\ell}=\langle\ell,\Phi(X)\rangle_{\mathcal H_{\omega}}$. Under moment
bounds on $\E\|\Sig\|^{2}$ the covariance operator
$\Gamma=\E[(\Phi-\E\Phi)\otimes(\Phi-\E\Phi)]$ is trace-class, self-adjoint
and positive; the defect form is its coordinate expression, and the mean
signature $\mu=\E\Phi$ is the mean embedding of the path law in the sense of
kernel mean embeddings \citep{muandet2017}, with the signature kernel
\citep{kiraly2019,salvi2021} as the reproducing kernel.

\begin{theorem}[Existence and value; \textsc{Proved}]
\label{thm:fredholm}
Let $\Gamma$ be a positive self-adjoint trace-class operator on a separable
Hilbert space and $\gamma>0$. Then
\[
\sup_{\ell}\ \langle\ell,\mu\rangle-\tfrac{\gamma}{2}\langle\ell,\Gamma\ell\rangle
\;=\;
\begin{cases}
\dfrac{1}{2\gamma}\,\|\Gamma^{\dagger/2}\mu\|^{2} & \text{if }
\mu\in\operatorname{Range}(\Gamma^{1/2}),\\[1ex]
+\infty & \text{otherwise,}
\end{cases}
\]
and when the supremum is finite it is attained if and only if
$\mu\in\operatorname{Range}(\Gamma)$, with maximizer
$\ell^{*}=\gamma^{-1}\Gamma^{\dagger}\mu$.
\end{theorem}

\begin{proof}
If $\mu$ has a component $\mu_{0}\neq0$ in $\ker\Gamma$, then $\ell=t\mu_{0}$
makes the objective $t\|\mu_{0}\|^{2}$, unbounded. Otherwise diagonalize
$\Gamma=\sum_{k}\lambda_{k}\,e_{k}\otimes e_{k}$ with $\lambda_{k}>0$ on
$(\ker\Gamma)^{\perp}$; the objective separates across coordinates, each
one-dimensional supremum is $\mu_{k}^{2}/(2\gamma\lambda_{k})$, and the total
is finite exactly when $\sum_{k}\mu_{k}^{2}/\lambda_{k}<\infty$, i.e.\
$\mu\in\operatorname{Range}(\Gamma^{1/2})$. The coordinatewise maximizer is
$\ell_{k}=\mu_{k}/(\gamma\lambda_{k})$, which lies in $\mathcal H_{\omega}$
exactly when $\sum_{k}\mu_{k}^{2}/\lambda_{k}^{2}<\infty$, i.e.\
$\mu\in\operatorname{Range}(\Gamma)$.
\end{proof}

In the diagonalization of the proof, the condition
$\mu\in\operatorname{Range}(\Gamma^{1/2})$ reads
$\sum_{k}\langle\mu,e_{k}\rangle^{2}/\lambda_{k}<\infty$: it is the
Picard condition of linear inverse problems \citep{engl1996}, with the
frontier slope of Corollary~\ref{cor:frontier} as the regularized solution
norm. Portfolio choice on path space is, in this exact sense, an inverse
problem: the market supplies the operator, the mean embedding supplies the
data, and asymptotic arbitrage is what ill-posedness looks like in finance.

The gap between the two range conditions is not a technicality. If
$\mu\in\operatorname{Range}(\Gamma^{1/2})\setminus\operatorname{Range}(\Gamma)$,
the value is finite but approached only along a sequence of ever-larger
controls: the strategy class contains approximate optimizers and no optimizer.
And if $\mu\notin\operatorname{Range}(\Gamma^{1/2})$, the class contains
payoff directions whose Sharpe ratio is unbounded --- an asymptotic arbitrage
inside the model. The condition
$\mu\in\operatorname{Range}(\Gamma^{1/2})$ is therefore the path-space
statement that the market prices the strategy class coherently.

\begin{corollary}[Signature efficient frontier; \textsc{Proved}]
\label{cor:frontier}
For $c>0$, $\max\{\langle\ell,\mu\rangle:\langle\ell,\Gamma\ell\rangle\le c\}
=\sqrt{c}\,\|\Gamma^{\dagger/2}\mu\|$ whenever the right side is finite. The
efficient frontier in $(\text{risk},\text{mean})$ coordinates is the ray of
slope $\|\Gamma^{\dagger/2}\mu\|$: the maximal Sharpe ratio of the class is
the norm whose finiteness Theorem~\ref{thm:fredholm} characterizes.
\end{corollary}

\subsection{The defect in log coordinates}

The slogan behind Corollary~\ref{cor:psd} --- risk is the failure of the mean
path to be a path --- admits an exact leading-order form. Group-like elements
are exponentials of Lie elements, and at level two the Lie part is
antisymmetric; so the symmetric level-two component of $\log\E[\Sig]$ is
precisely the lowest-order obstruction to group-likeness.

\begin{proposition}[Log-defect identity; \textsc{Proved}]
\label{prop:logdefect}
For any geometric lift with square-integrable level-two signature, the
symmetric part of the level-two component of $\log\E[\Sig(X)]$ equals
$\tfrac12\Cov(X_{T}-X_{0})$, which is one half the level-one block of the
defect form $D$.
\end{proposition}

\begin{proof}
Pathwise, the shuffle identity at level two gives
$\Sig^{ij}+\Sig^{ji}=\Sig^{i}\Sig^{j}$, so
$\operatorname{sym}\E[\Sig^{(2)}]=\tfrac12\E[\Delta X\otimes\Delta X]$ with
$\Delta X=X_{T}-X_{0}$. The tensor logarithm at level two is
$\E[\Sig^{(2)}]-\tfrac12\,\E[\Delta X]\otimes\E[\Delta X]$, whose symmetric
part is $\tfrac12(\E[\Delta X\otimes\Delta X]-\E[\Delta X]\otimes\E[\Delta X])
=\tfrac12\Cov(\Delta X)$. The level-one block of $D$ is
$D((i),(j))=\E[\Sig^{ij}+\Sig^{ji}]-\E[\Sig^{i}]\E[\Sig^{j}]
=\Cov(\Delta X)_{ij}$.
\end{proof}

So in log coordinates the leading non-Lie component of the expected signature
\emph{is} the covariance of terminal increments, halved: risk is literally the
first thing one meets on the way from $\E[\Sig]$ back to the group. The
identity is distribution-free, and holds numerically on a two-point-mixture
driver to $1.1\times10^{-7}$ on objects of size $4.5\times10^{-2}$.

\begin{corollary}[Distance to geometricity in the log chart; \textsc{Proved}]
\label{cor:distance}
At truncation level two, the Lie algebra
$\mathfrak g_{2}=\{(a,A):a\in\R^{d},\,A\text{ antisymmetric}\}$ is a
linear subspace of log coordinates, and the Euclidean distance from
$\log\E[\Sig]$ to $\mathfrak g_{2}$ is exactly
$\tfrac12\|\Cov(X_{T}-X_{0})\|_{F}$. The distance from the expected
signature to the group-like variety, measured in the log chart, is half the
Frobenius norm of the covariance of terminal increments.
\end{corollary}

\begin{proof}
Orthogonal projection onto $\mathfrak g_{2}$ removes exactly the symmetric
level-two block of the log; Proposition~\ref{prop:logdefect} identifies that
block as $\tfrac12\Cov(X_{T}-X_{0})$.
\end{proof}

Group-like elements form a variety, not a subspace, so ``distance to
geometricity'' needs a chart to be well defined; the log chart is the natural
one because it linearizes the variety, and in it the distance is not merely
bounded by the covariance --- it \emph{is} the covariance, halved. The
projection onto $\mathfrak g_{2}$ used in the proof is, at level two, the
first Eulerian idempotent of the tensor algebra
\citep{reutenauer1993}, so the corollary sits inside the canonical
Lie-theoretic decomposition rather than being an ad hoc matrix split; at
higher truncation levels the same idempotent defines the projection, and the
resulting distance formula is open.

\subsection{Log utility and the Kelly ray}

\begin{proposition}[Second-order Kelly control; \textsc{Proved}]
\label{prop:kelly}
Let terminal wealth be $1+G_{\ell}$. Then
\[
\E\log(1+G_{\ell})
=\langle\ell,\mu\rangle
-\tfrac12\big(\langle\ell,\Gamma\ell\rangle
+\langle\ell,\mu\rangle^{2}\big)+O(\|\ell\|^{3}),
\]
and the second-order log-optimal control is
\[
\ell_{K}=(\Gamma+\mu\otimes\mu)^{\dagger}\mu
=\frac{\Gamma^{\dagger}\mu}{1+\langle\mu,\Gamma^{\dagger}\mu\rangle}.
\]
The Kelly control lies on the same ray as every mean--variance optimizer ---
the efficient-frontier ray of Corollary~\ref{cor:frontier} --- shrunk by one
plus the squared maximal Sharpe ratio of the class.
\end{proposition}

\begin{proof}
$\log(1+x)=x-\tfrac12x^{2}+O(x^{3})$ and
$\E[G_{\ell}^{2}]=\langle\ell,\Gamma\ell\rangle
+\langle\ell,\mu\rangle^{2}$; the second display is Sherman--Morrison.
Dropping the $\langle\ell,\mu\rangle^{2}$ term would return the
$\gamma=1$ mean--variance control exactly; the correct expansion keeps the
direction and shrinks the scale.
\end{proof}

The expansion is formal in $\|\ell\|$: by Corollary~\ref{cor:ruin} the
argument $1+G_{\ell}$ need not be positive for every control under the
forward lift, so the statement is to be read on a neighbourhood of $\ell=0$,
or under a geometric lift where wealth is an exponential and the logarithm is
never evaluated off its domain.

So the Kelly criterion, Markowitz, and signature optimization coincide in
direction on path space and differ only in how far along the frontier ray they
sit; the universal-portfolio reading of the same ray is \citet{cover1991}, and
exact (rather than second-order) log-optimal signature portfolios are the
object approximated in \citet{cuchiero2023moller}.

\section{The lift is the execution convention}
\label{sec:lift}

\begin{theorem}[Lift gap at level two; \textsc{Proved}]
\label{thm:lift}
With $\Sig^{M}$ the Marcus (geometric) and $\Sig^{F}$ the forward
iterated-sums level-two signature of the same stream,
$\Sig^{M,ij}-\Sig^{F,ij} = \tfrac12[X^{i},X^{j}]_{T}$: the lifts differ only in
their symmetric block, by one half the quadratic covariation.
\end{theorem}

For continuous semimartingales this is the It\^o--Stratonovich correction; on
pure-jump streams it is Hoffman's exponential
\citep{hoffman2000,marcus1981}, the algebraic content of the Hopf square of
\citet{noguer2026paths}. Simulation confirms both: with $20{,}000$ paths of
$1000$ steps the mean level-two difference is
$\begin{psmallmatrix}0.020011&0.006002\\0.006002&0.044993\end{psmallmatrix}$
against $\tfrac12\Sigma T=
\begin{psmallmatrix}0.020&0.006\\0.006&0.045\end{psmallmatrix}$, maximum error
$10^{-5}$; on a pure-jump stream
$\Sig^{M}-\Sig^{F}=\tfrac12\sum\Delta X\,\Delta X^{\top}$ holds to
$1.7\times10^{-15}$.

\begin{corollary}[Excess growth is the geometricity defect; \textsc{Proved}]
\label{cor:excess}
For weights $\pi$ held fixed across a log-price move $u$,
\[
\underbrace{\pi^{\top}u}_{\text{level-one signature coordinate}}
-\underbrace{\log\!\big(1+\pi^{\top}(e^{u}-1)\big)}_{\text{one-period portfolio return}}
= -\tfrac12\Big(\textstyle\sum_{i}\pi_{i}u_{i}^{2}-(\pi^{\top}u)^{2}\Big)
+O(|u|^{3}),
\]
whose leading term is $-\gamma^{*}_{\pi}$ for the rank-one covariance
$uu^{\top}$. In the continuous limit
$\log V_{T}-\int_{0}^{T}\pi^{\top}d\log X=\int_{0}^{T}\gamma^{*}_{\pi}dt$ with
$\gamma^{*}_{\pi}=\tfrac12(\pi^{\top}\mathrm{diag}\,\Sigma-\pi^{\top}\Sigma\pi)$.
\end{corollary}

\begin{proof}
For the discrete statement,
$\varphi(u)=\pi^{\top}u-\log(1+\pi^{\top}(e^{u}-1))$ satisfies $\varphi(0)=0$,
$\nabla\varphi(0)=0$ and
$\varphi(\varepsilon v)=-\tfrac{\varepsilon^{2}}{2}(\sum_{i}\pi_{i}v_{i}^{2}
-(\pi^{\top}v)^{2})+O(\varepsilon^{3})$; the bracket is $2\gamma^{*}_{\pi}$ for
$vv^{\top}$.

For the continuous statement, let $S^{i}=e^{X^{i}}$ with $X$ a continuous
semimartingale, and let $V$ be the self-financing wealth of the constant-weight
portfolio, $dV_{t}/V_{t}=\sum_{i}\pi_{i}\,dS^{i}_{t}/S^{i}_{t}$. It\^o's
formula gives $dS^{i}/S^{i}=dX^{i}+\tfrac12\,d[X^{i},X^{i}]$ and
$d\log V = dV/V - \tfrac12\,d[V]/V^{2}$, whence
\[
d\log V_{t}
=\pi^{\top}dX_{t}
+\tfrac12\Big(\sum_{i}\pi_{i}\,d[X^{i},X^{i}]_{t}
-\sum_{i,j}\pi_{i}\pi_{j}\,d[X^{i},X^{j}]_{t}\Big)
=\pi^{\top}dX_{t}+\gamma^{*}_{\pi}\,dt .
\]
The first term integrates to $\pi^{\top}(X_{T}-X_{0})$, the level-one
signature coordinate, which is identical under both lifts. The second term is
exactly the contraction of the level-two lift gap of Theorem~\ref{thm:lift}
with the tensor $\operatorname{diag}(\pi)-\pi\pi^{\top}$:
\[
\int_{0}^{T}\gamma^{*}_{\pi}\,dt
=\big\langle \tfrac12[X,X]_{T},\;
\operatorname{diag}(\pi)-\pi\pi^{\top}\big\rangle
=\big\langle \Sig^{M}-\Sig^{F},\;
\operatorname{diag}(\pi)-\pi\pi^{\top}\big\rangle_{\mathrm{lvl\,2}} .
\]
So the entire difference between logarithmic wealth and the level-one signature
flow is the lift gap read through one fixed portfolio tensor; the map
$\pi\mapsto\operatorname{diag}(\pi)-\pi\pi^{\top}$ is how the bracket
obstruction enters the one-dimensional wealth process.
\end{proof}

Both displays put $\gamma^{*}$ on the same side. The one-period portfolio
return exceeds the level-one coordinate by $\gamma^{*}$, and continuously
rebalanced wealth exceeds $\int\pi^{\top}d\log X$ by $\int\gamma^{*}dt$: what
carries no rebalancing premium is the level-one signature flow itself, not
either wealth process. The elementary check is $\pi=(\tfrac12,\tfrac12)$ and
$u=(a,-a)$, where the level-one coordinate is zero and the executed return is
$\log\cosh a\approx a^{2}/2=\gamma^{*}$.

The continuous identity is classical \citep{fernholz2002,fernholzkaratzas2009};
the placement is the point. Excess growth is not a rebalancing artifact that
stochastic portfolio theory happens to isolate --- it \emph{is} the
geometricity defect of the portfolio map, the same bracket obstruction as in
Theorem~\ref{thm:lift}. Signatures have been brought into stochastic portfolio
theory before, as a universal parameterization of path-dependent portfolio
\emph{weights} \citep{cuchiero2023moller}; the identification made here runs in
the other direction, reading SPT's central scalar as a statement about the lift
of the signature itself, and to our knowledge it has not been stated.

\begin{corollary}[Direction is lift-invariant; \textsc{Proved}]
\label{cor:area}
The antisymmetric part of the level-two signature is identical under both
lifts, pathwise; measured discrepancy $6.5\times10^{-16}$.
\end{corollary}

Arguments about It\^o versus Stratonovich versus Marcus can change the sign and
size of a variance-trading or rebalancing P\&L. They cannot touch a lead--lag
estimate.

\begin{corollary}[Ruin is a convention; \textsc{Proved}]
\label{cor:ruin}
Marcus-lifted wealth $\exp(\pi^{\top}u)$ is strictly positive for every bounded
control and finite move; forward-lifted wealth $1+\pi^{\top}(e^{u}-1)$ is
non-positive whenever $\pi^{\top}(e^{u}-1)\le-1$, attainable at one jump for
levered $\pi$.
\end{corollary}

A levered portfolio rebalanced continuously through a crash cannot go bankrupt
at a jump, because it is always selling into the fall; executed at the jump it
can. Any risk system built on a geometric lift has silently assumed away the
event it most needs to measure.

\section{The price of path complexity}
\label{sec:cost}

Everything so far assumes $\E[\Sig]$ is known. It is not. Because
Proposition~\ref{prop:defect} makes $D$ a covariance of $p=d+d^{2}$
payoffs, the plug-in policy $\hat\ell = (\gamma\hat D)^{-1}\hat\mu$ inherits
the pathology that \citet{michaud1989} and \citet{demiguel2009} document for
mean--variance in the static case --- with the difference that $p$ now grows
quadratically in $d$. The experiments below estimate $\hat\mu$ and $\hat D$
from $M$ independent paths, form $\hat\ell$, and evaluate its \emph{true}
certainty equivalent $\CE(\hat\ell)=\mu^{\top}\hat\ell
-\tfrac{\gamma}{2}\hat\ell^{\top}D\hat\ell$ against the exact $(\mu,D)$.

\paragraph{Calibrations.} Both drivers are $X_{t}=bt+\sigma W_{t}$ on $[0,1]$
under a geometric lift at $\gamma=3$, so $\E[\Sig]=\exp(T(b+\tfrac12\Sigma))$
is exact and $D$ is assembled from it by shuffling. At $d=2$, the calibration
of Section~\ref{sec:check}: $p=6$, $\operatorname{cond}D=134.3$. At $d=20$,
volatilities equally spaced on $[0.22,0.35]$, equicorrelation $0.40$, drifts
equally spaced on $[0.04,0.10]$: $p=420$, $\operatorname{cond}D=1906$ with
$81.0\%$ of the eigenvalues above $10^{-3}\lambda_{\max}$. The population
problem is well posed at both sizes, so everything below is a sample-size
statement and not a conditioning statement.

\paragraph{Four policies.} \emph{Raw} is the unregularized plug-in.
\emph{Ridge} replaces $\hat D$ by $\hat D+\delta\bar\lambda I$ with
$\bar\lambda=\operatorname{tr}\hat D/p$, at two intensities: the fixed
$\delta=0.25$ of identity-target shrinkage \citep{ledoit2004}, and
$\delta=\sqrt{p/M}$, the ambiguity radius that
Proposition~\ref{prop:robust} and the Marchenko--Pastur discussion below both
prescribe. \emph{Model} is a model-consistent plug-in: it estimates only the
generator $(\hat b,\hat\Sigma)$ from the level-one increments and rebuilds
$\hat\mu$ and $\hat D$ from it by tensor exponential and shuffle. The
level-one plug-in is carried throughout as a benchmark. Cells report medians
over independent replications, with interquartile ranges and $\Pr[\CE<0]$,
because in the regime that matters the loss is heavy-tailed and its mean is
not a stable summary.

Before the tables, one identity organizes everything in them.

\begin{proposition}[Exact loss decomposition; \textsc{Proved}]
\label{prop:celoss}
For any $\ell$,
$\CE^{*}-\CE(\ell)=\tfrac{\gamma}{2}\,\|D^{1/2}(\ell-\ell^{*})\|^{2}$.
In particular
$\E[\CE(\hat\ell)]=\CE^{*}
-\tfrac{\gamma}{2}\,\E\|D^{1/2}(\hat\ell-\ell^{*})\|^{2}$ for any
estimator $\hat\ell$: the expected plug-in loss is the $D$-weighted mean
squared error of the weights, with no higher-order terms.
\end{proposition}

\begin{proof}
$\CE$ is a concave quadratic with maximizer $\ell^{*}$; expand around it.
\end{proof}

The identity is why the tables behave as they do. The plug-in error
$\hat\ell-\ell^{*}$ contains $\hat D^{-1}$, so its $D$-weighted second moment
inherits the spectrum of the inverse sample covariance. Two facts locate the
barrier without any distributional assumption. First,
$\operatorname{rank}(\hat D)\le M-1$ exactly, so $M\ge p+1$ is necessary for
the raw plug-in to be defined at all --- linear algebra, not asymptotics.
Second, for pure noise with finite fourth moments the sample spectrum
concentrates on the Marchenko--Pastur bulk
$[(1-\sqrt{p/M})^{2},(1+\sqrt{p/M})^{2}]$ \citep{marchenko1967}, so the
smallest eigenvalue collapses as $M\downarrow p$ and $\|\hat D^{-1}\|$
diverges together with the loss of Proposition~\ref{prop:celoss}; the measured
edge tracks the bound across $M/p$ from $1.3$ to $200$. Signature coordinates are neither independent nor
identically distributed across levels, so Marchenko--Pastur is a located
analogy, not a theorem about $\hat D$; the rank statement needs no such
caveat. The random-matrix toolkit for cleaning sample covariance matrices in
finance is surveyed by \citet{bun2017}; what is specific here is that the
payoffs whose covariance is being cleaned are generated algebraically, by
words, rather than observed, so their number is chosen by the truncation level
and not by the data.

The exact finite-sample theory of \citet{kanzhou2007} is the natural thing to
reach for here, and it is worth being precise about how far it reaches. It is
exact for the level-one sub-problem, where the payoffs are Gaussian and
$\hat D$ is Wishart, and there it predicts the divergence of the expected
plug-in loss as $M\downarrow p$. It does not transfer to level two: those
payoffs are quadratic in Gaussians, $\hat D$ is not Wishart, and we have no
closed form for the negative entries of Table~\ref{tab:est20}. Those entries
are measured, not derived, which is why they are reported with dispersion.

The rate itself can be stated as a theorem, with hypotheses the signature
setting makes natural --- and with the dimension entering through the trace
and the conditioning of $D$ rather than through the word count.

\begin{theorem}[Consistency threshold; \textsc{Proved} under (i)--(ii)]
\label{thm:consistency}
Assume (i) the level-$\le m$ signature coordinates are bounded,
$\|\Phi\|\le B$ almost surely --- automatic on paths of $1$-variation at
most $R$, since $|\ip{w}{\Sig}|\le R^{|w|}/|w|!$ --- and (ii)
$\lambda_{\min}(D)\ge\lambda_{0}>0$ on the tradable span. Let $V$ be the
matrix variance proxy of the summands of $\hat D$ and
$\operatorname{intdim}(V)=\operatorname{tr}V/\|V\|$. Then there are constants
$C_{0},C_{1}$ depending only on $(B,\lambda_{0},\gamma,\|\ell^{*}\|)$ such
that for $M\ge C_{0}\log(\operatorname{intdim}(V)/\delta)$, with probability
at least $1-\delta$,
\[
\CE^{*}-\CE(\hat\ell)\;\le\;C_{1}\,
\frac{\operatorname{tr}D+\log(1/\delta)}{M}.
\]
The word count $p$ appears nowhere in the display. Bounding
$\operatorname{tr}D\le B^{2}p$ recovers the familiar $p/M$ rate, but
Remark~\ref{rem:effdim} shows that bound to be wasteful exactly in the
signature setting; the operative ratio is
$\operatorname{tr}(D)/(\lambda_{0}M)$.
\end{theorem}

\begin{proof}
By Proposition~\ref{prop:celoss} it suffices to bound
$\|D^{1/2}(\hat\ell-\ell^{*})\|$. With $\hat D$ centred at the population
mean the identity
$\hat\ell-\ell^{*}=\gamma^{-1}\hat D^{-1}\big[(\hat\mu-\mu)
-\gamma(\hat D-D)\ell^{*}\big]$ is exact; sample centring perturbs $\hat D$
by $(\hat\mu-\mu)(\hat\mu-\mu)^{\top}$, of operator norm $O_{P}(1/M)$, which
contributes below the stated order. The intrinsic-dimension form of matrix
Bernstein \citep{tropp2012,tropp2015} gives
$\|\hat D-D\|_{\mathrm{op}}\le\lambda_{0}/2$ with probability
$1-\delta/2$ once $M\ge C_{0}\log(\operatorname{intdim}(V)/\delta)$, and on
that event $\|\hat D^{-1}\|\le2/\lambda_{0}$. Vector Bernstein gives
$\|\hat\mu-\mu\|^{2}\lesssim(\operatorname{tr}D+B^{2}\log(1/\delta))/M$, and
matrix Bernstein gives $\|(\hat D-D)\ell^{*}\|^{2}\lesssim
B^{4}\|\ell^{*}\|^{2}(1+\log(1/\delta))/M$, which is dimension-free.
Combining terms gives the display; the $\hat\mu$ term carries the trace.
\end{proof}

The theorem is an upper bound under hypothesis (ii), which is a statement
about $D$. What fails in the tables is $\lambda_{\min}(\hat D)$, not
$\lambda_{\min}(D)$: the good event of the proof is exactly what a sample of
size $M\approx p$ does not deliver. So the theorem and the tables are not two
readings of one quantity --- the first bounds the loss once the sample
resolves the spectrum, the second measures what happens before it does.

\begin{remark}[Effective dimension, and where the difficulty actually lives]
\label{rem:effdim}
On paths of $1$-variation at most $R$ the factorial decay
$|\ip{w}{\Sig}|\le R^{|w|}/|w|!$ gives
$\operatorname{tr}(D)\le\sum_{k\le m}d^{k}R^{2k}/(k!)^{2}$, a series that
converges even as $m\to\infty$: the trace side of the sample-size floor does
not grow with the truncation level at all, and neither does the intrinsic
dimension, which is $1.49$ at $d=2$ and $3.84$ at $d=20$ here against word
counts of $6$ and $420$. The difficulty migrates into the conditioning ---
$\lambda_{\min}(D)$ decays factorially across levels and the constants carry
$\lambda_{0}^{-2}$ --- so what the trace saves, the spectrum spends. The
operative quantity is the effective dimension
$d_{\mathrm{eff}}=\operatorname{tr}(D)/\lambda_{\min}(D)$, equal to $200$ at
$d=2$ and $7325$ at $d=20$. It summarizes the tables better than $p$ does:
the raw plug-in reaches $85\%$ of its oracle value at $M=0.62\,
d_{\mathrm{eff}}$ in the pair and $M=0.68\,d_{\mathrm{eff}}$ in the cross
section, whereas the corresponding values of $M/p$, namely $20.8$ and
$11.9$, differ by a factor of $1.8$. This is a measured statement about two
calibrations, not a theorem.
\end{remark}

\begin{remark}[Bias--variance across truncation levels]
The oracle value $\CE^{*}_{m}$ is nondecreasing in $m$; write
$A_{m}=\lim_{k}\CE^{*}_{k}-\CE^{*}_{m}\ge0$ for the approximation gap of
level $m$. Theorem~\ref{thm:consistency} then gives, with high probability,
$\CE(\hat\ell_{m})\ge\lim_{k}\CE^{*}_{k}-A_{m}
-C_{1}\operatorname{tr}(D_{m})/M$: the optimal truncation balances a gap that
universality drives to zero against a cost that grows with the effective
dimension of level $m$. We attach no rate to $A_{m}$ --- none holds without
regularity assumptions on the path law --- so the balance is stated, not
solved.
\end{remark}

\begin{table}[ht]
\centering\small
\begin{tabular}{rrrrrrrrr}
\toprule
& & \multicolumn{3}{c}{raw plug-in} & \multicolumn{2}{c}{ridge, frac.\ oracle}
& model & level $1$\\
\cmidrule(lr){3-5}\cmidrule(lr){6-7}\cmidrule(lr){8-8}\cmidrule(lr){9-9}
$M$ & $M/p$ & med.\ $\CE$ & frac. & $\Pr[\CE<0]$
& $\delta=0.25$ & $\delta=\sqrt{p/M}$ & frac. & frac.\\
\midrule
$60$   & $10.0$   & $0.1001$ & $0.598$ & $0.175$ & $0.485$ & $0.436$ & $0.957$ & $0.706$\\
$125$  & $20.8$   & $0.1428$ & $0.853$ & $0.005$ & $0.493$ & $0.521$ & $0.979$ & $0.863$\\
$250$  & $41.7$   & $0.1562$ & $0.933$ & $0.000$ & $0.500$ & $0.601$ & $0.991$ & $0.922$\\
$500$  & $83.3$   & $0.1619$ & $0.968$ & $0.000$ & $0.501$ & $0.672$ & $0.996$ & $0.967$\\
$1000$ & $166.7$  & $0.1646$ & $0.983$ & $0.000$ & $0.503$ & $0.740$ & $0.998$ & $0.984$\\
$2000$ & $333.3$  & $0.1660$ & $0.992$ & $0.000$ & $0.502$ & $0.799$ & $0.999$ & $0.992$\\
$8000$ & $1333.3$ & $0.1670$ & $0.998$ & $0.000$ & $0.503$ & $0.895$ & $1.000$ & $0.998$\\
\bottomrule
\end{tabular}
\caption{$d=2$, $p=6$, $\gamma=3$, $d_{\mathrm{eff}}=200$. Oracle certainty
equivalents $0.167342$ at level $\le2$ and $0.015008$ at level $1$, ratio
$11.150$. Medians over $400$ replications ($200$ at $M\ge2000$); the level-$1$
column is a fraction of the level-$1$ oracle, all others of the level-$\le2$
oracle. The raw interquartile range is $[0.040,0.133]$ at $M=60$ and inside
$\pm4\%$ of the median from $M=250$ on.}
\label{tab:est2}
\end{table}

\begin{table}[ht]
\centering\footnotesize\setlength{\tabcolsep}{4pt}
\begin{tabular}{rrrrrrrrrr}
\toprule
& & \multicolumn{4}{c}{raw plug-in} & \multicolumn{2}{c}{ridge, frac.\ oracle}
& model & level $1$\\
\cmidrule(lr){3-6}\cmidrule(lr){7-8}\cmidrule(lr){9-9}\cmidrule(lr){10-10}
$M$ & $M/p$ & med.\ $\CE$ & IQR & frac. & $\Pr[\CE<0]$
& $\delta=0.25$ & $\delta=\sqrt{p/M}$ & frac. & frac.\\
\midrule
$500$   & $1.19$  & $-510.9$ & $[-634,-402]$   & $-306.1$ & $1.00$ & $0.577$ & $0.301$ & $0.950$ & $0.690$\\
$1000$  & $2.38$  & $-3.857$ & $[-4.32,-3.34]$ & $-2.311$ & $1.00$ & $0.618$ & $0.385$ & $0.978$ & $0.849$\\
$2000$  & $4.76$  & $0.675$  & $[0.587,0.756]$ & $0.405$  & $0.00$ & $0.627$ & $0.472$ & $0.989$ & $0.930$\\
$3000$  & $7.14$  & $1.163$  & $[1.119,1.203]$ & $0.697$  & $0.00$ & $0.627$ & $0.524$ & $0.993$ & $0.952$\\
$5000$  & $11.90$ & $1.436$  & $[1.416,1.457]$ & $0.861$  & $0.00$ & $0.629$ & $0.591$ & $0.996$ & $0.972$\\
$20000$ & $47.62$ & $1.623$  & $[1.622,1.625]$ & $0.973$  & $0.00$ & $0.628$ & $0.756$ & $0.999$ & $0.993$\\
\bottomrule
\end{tabular}
\caption{$d=20$, $p=420$, $\gamma=3$, $d_{\mathrm{eff}}=7325$. Oracle
certainty equivalents $1.669016$ at level $\le2$ and $0.027984$ at level $1$,
ratio $59.64$. Medians and interquartile ranges over $200,200,120,100,60,20$
replications by row; the level-$1$ column is a fraction of the level-$1$
oracle, all others of the level-$\le2$ oracle. At $M/p=1.19$ the whole
interquartile range is negative, so the sign is not a tail artifact.}
\label{tab:est20}
\end{table}

\begin{figure}[ht]
\centering
\input{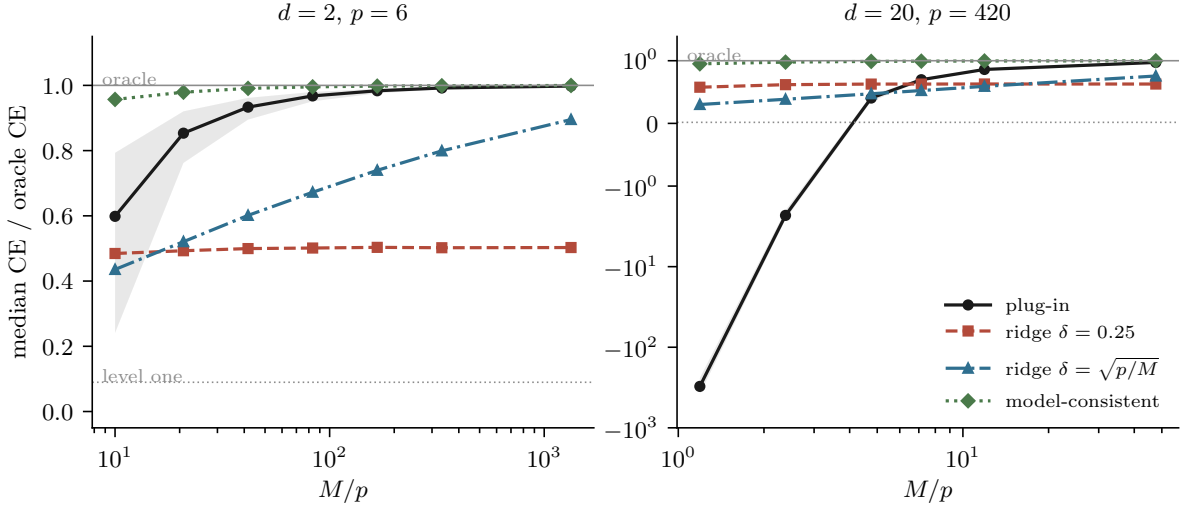}
\caption{Median realized certainty equivalent as a fraction of the oracle,
against sample size in units of the parameter count, from the runs behind
Tables~\ref{tab:est2} and~\ref{tab:est20}. Left: the pair. Right: the cross
section, on a symmetric-logarithmic scale because the unregularized plug-in is
two orders of magnitude below zero at $M/p=1.19$. Shading is the interquartile
range of the plug-in. Grey lines mark the oracle and the level-one value. The
two ridge intensities separate cleanly: the fixed one is flat and never
converges, the Marchenko--Pastur radius does; the model-consistent estimator
starts where the others end up.}
\label{fig:est}
\end{figure}

\begin{proposition}[Sample-size floor for path complexity;
\textsc{Conditional, measured}]
\label{prop:floor}
On the calibrations of Tables~\ref{tab:est2}--\ref{tab:est20}, the raw
level-$\le2$ plug-in crosses half of its oracle value between $M/p=4.8$ and
$M/p=7.1$, and at $M/p\le2.4$ its realized certainty equivalent is negative
in every replication --- at $M/p=1.2$ the median is $-306\times$ the oracle
value it was chasing, with an interquartile range of
$[-634,-402]$ in absolute terms. The static plug-in is essentially unaffected
across the same range.
\end{proposition}

\begin{proposition}[Shrinkage changes sign with dimension;
\textsc{Conditional, measured}]
\label{prop:shrink}
The fixed intensity $\delta = 0.25$ destroys half the attainable value at
$p=6$ ($0.49$--$0.50$ of oracle at every sample size, against $0.998$ for the
raw plug-in at $M=8000$) and is the difference between $-306\times$ and
$+0.58\times$ the oracle at $p=420$, $M=500$. Its virtue is that it is flat
and its defect is that it is inconsistent: it delivers $0.58$--$0.63$ of
oracle at every sample size tested at $p=420$ and never converges. The
Marchenko--Pastur intensity $\delta=\sqrt{p/M}$ is the reverse --- worse at
the smallest samples ($0.30$ against $0.58$ at $M/p=1.19$) and the only ridge
that converges, rising monotonically to $0.90$ at $p=6$ and $0.76$ at $p=420$
over the range tested.
\end{proposition}

Three readings, in increasing order of consequence.

First, the oracle case gets \emph{better} with dimension: the advantage of
admitting quadratic path functionals rises from $11.15\times$ at $d=2$ to
$59.64\times$ at $d=20$. But it is worth knowing which words carry it.

\begin{table}[ht]
\centering
\begin{tabular}{lrrrr}
\toprule
& \multicolumn{2}{c}{$d=2$} & \multicolumn{2}{c}{$d=20$}\\
\cmidrule(lr){2-3}\cmidrule(lr){4-5}
basis & $\CE^{*}$ & $\times$ level $1$ & $\CE^{*}$ & $\times$ level $1$\\
\midrule
level $1$                              & $0.015008$ & $1.00$  & $0.027984$ & $1.00$\\
$+$ antisymmetric                      & $0.015008$ & $1.00$  & $0.027984$ & $1.00$\\
$+$ symmetric off-diagonal             & $0.020200$ & $1.35$  & $0.081437$ & $2.91$\\
$+$ symmetric diagonal                 & $0.160804$ & $10.71$ & $0.393548$ & $14.06$\\
$+$ symmetric block $=$ level $\le2$   & $0.167342$ & $11.15$ & $1.669016$ & $59.64$\\
\bottomrule
\end{tabular}
\caption{Oracle certainty equivalent on sub-bases of the level-$\le2$ word
basis, $\gamma=3$. Each row adds the named class to level one; the last row
is the full level-$\le2$ problem.}
\label{tab:wordclass}
\end{table}

\begin{proposition}[Where the level-two premium lives;
\textsc{Conditional, measured}]
\label{prop:wordclass}
On both calibrations the entire level-$\le2$ oracle gain sits in the
symmetric block. The antisymmetric block adds exactly nothing: under these
drivers its mean vanishes and it is uncorrelated with every other
level-$\le2$ word, both to machine precision, so the level-one and
level-one-plus-area values agree to sixteen digits. Within the symmetric
block the diagonal words carry almost all of the pair's gain
($10.71\times$ of $11.15\times$) and the two halves are strongly
super-additive in the cross section: $14.06\times$ from the diagonal words
alone, $2.91\times$ from the off-diagonal words alone, $59.64\times$
together.
\end{proposition}

\begin{figure}[ht]
\centering
\input{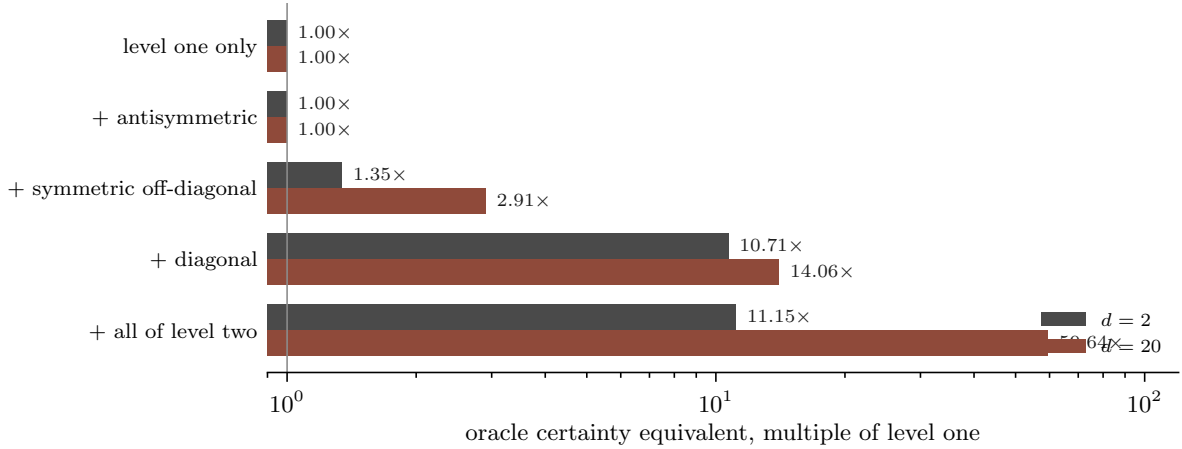}
\caption{Oracle certainty equivalent of each level-$\le2$ sub-basis, as a
multiple of the level-one value, on the two calibrations. The antisymmetric
words add exactly nothing at either cross-sectional width --- the bar is
$1.00\times$ to sixteen digits --- while the diagonal words carry most of the
gain and the two symmetric halves are strongly super-additive at $d=20$.}
\label{fig:wordclass}
\end{figure}

This changes what the dimensional trade-off is a statement about. Under a
geometric lift the shuffle relation makes the symmetric level-two block
exactly $\tfrac12\,\Delta X\otimes\Delta X$, a quadratic function of the
terminal increment; only the antisymmetric block --- the L\'evy area --- is
genuinely path-dependent. So what level two buys on these drivers is
convexity in the terminal return, and the path-dependent half of the level
earns nothing here, by construction: the driver has no expected area.
Section~\ref{sec:area} is where the antisymmetric words are put in front of a
driver that pays them. That the antisymmetric block earns exactly zero
here is a statement about the driver --- which is time-reversible and has no
expected area --- and not a statement about markets. And by Proposition~\ref{prop:complement} the genuine
variance swap $[X,X]_{T}$ is not in this universe at all --- it enters only
under the forward lift. The honest name for the $59.64\times$ is therefore the
price of admitting quadratic payoffs in the terminal increment. The cross
section beats the pair because it has $d(d+1)/2$ of them rather than three ---
though only just: the count rises from three to two hundred and ten while the
value rises tenfold, so on equicorrelated assets the marginal quadratic
direction is worth very little, and it is the estimation cost that scales with
the count.

Second, the estimation problem gets worse faster. The oracle advantage grows
roughly linearly in $d$ here while $p$ grows quadratically, so the sample size
required to realize the advantage outruns the advantage itself. This is a
tradability statement: path complexity survives only where the sample does.

Third --- and this is the practical inversion --- how much to regularize is
not a matter of taste. At $d=2$ a practitioner who shrinks at fixed intensity
is throwing away half the value. At $d=20$ the same practitioner is the only
one still solvent. Conditioning is not the culprit: the population problem is
well posed at both sizes, as the calibration paragraph records. The failure is
a sample-size failure, which is why shrinkage --- and not a better inversion
--- is what repairs it.

\subsection{The floor is a property of the estimator}

It is tempting to look for the repair in the algebra rather than in
shrinkage, since the shuffle relations are exact constraints linking the
coordinates of $\E[\Sig]$. They give nothing here, and the reason is worth
stating: those relations hold \emph{pathwise}, so any sample-based
$\hat D$ satisfies every one of them identically. There is no violated
constraint to project onto. What the algebra does supply is an assembly map
--- from a generator to $\E[\Sig]$ by tensor exponential, and from
$\E[\Sig]$ to $D$ by shuffling --- and that map becomes an estimator as soon
as one is willing to name a driver class.

\begin{proposition}[Model-consistent estimation removes the floor;
\textsc{Conditional, measured}]
\label{prop:model}
Estimating only $(\hat b,\hat\Sigma)$ from the level-one increments ---
$d+d(d+1)/2$ numbers, $230$ at $d=20$ against $420$ free means and
$88{,}410$ free covariance entries --- and rebuilding
$\hat\mu=\exp(T(\hat b+\tfrac12\hat\Sigma))$ and $\hat D$ from it by
shuffling recovers $0.950$ of the oracle at $M/p=1.19$ and $0.999$ at
$M/p=47.6$, where the raw plug-in reads $-306\times$ and $0.973$. At $d=2$ it
recovers $0.957$ at $M=60$. The sample-size floor of
Proposition~\ref{prop:floor} is therefore a property of unstructured
estimation of the risk form, not of path complexity.
\end{proposition}

The caveat is exactly as large as the claim. The model-consistent estimator
is evaluated under a correctly specified driver, so what it measures is the
cost of not using a model, not the cost of using the wrong one; under
misspecification its bias does not vanish with $M$ and the comparison would
invert at some sample size we do not identify. What survives the caveat is
the decomposition: the $-306\times$ of Table~\ref{tab:est20} is the price of
estimating $420\times421/2$ second moments freely, and it is not a statement
about level-two words as such.

The ridge itself deserves a better pedigree than ``the fix that works''.

\begin{proposition}[Shrinkage is priced ambiguity; \textsc{Proved}]
\label{prop:robust}
Consider the robust problem
$\max_{\ell}\min\{\,\ell^{\top}\tilde\mu
-\tfrac{\gamma}{2}\ell^{\top}\tilde D\ell\;:\;
\|\tilde\mu-\hat\mu\|\le r_{\mu},\
\|\tilde D-\hat D\|_{\mathrm{op}}\le r_{D}\,\}$.
The inner minimum equals
$\ell^{\top}\hat\mu-r_{\mu}\|\ell\|
-\tfrac{\gamma}{2}\ell^{\top}(\hat D+r_{D}I)\ell$. In particular, with
$r_{\mu}=0$ the robust optimizer is \emph{exactly} the ridge policy
$\ell_{R}=\big(\gamma(\hat D+r_{D}I)\big)^{-1}\hat\mu$, and a mean
radius $r_{\mu}>0$ adds a norm penalty on top.
\end{proposition}

\begin{proof}
$\sup_{\|E\|_{\mathrm{op}}\le r_{D}}\ell^{\top}E\ell
=r_{D}\|\ell\|^{2}$, attained at $E=r_{D}\,\ell\ell^{\top}/\|\ell\|^{2}$,
and $\inf_{\|e\|\le r_{\mu}}\ell^{\top}e=-r_{\mu}\|\ell\|$.
\end{proof}

This closes the loop with the tables, and the two ridge columns are the two
sides of it. The intensity $\delta=\sqrt{p/M}$ \emph{is} the
covariance-ambiguity radius that the Marchenko--Pastur discussion prescribes
for sampling noise, and it behaves as a priced ambiguity should: it shrinks
away as the sample grows, so the policy converges. The fixed $\delta=0.25$ is
not a radius at all --- it prices an ambiguity that does not shrink, which is
why it is flat in every row of both tables, better than the radius when the
sample is smallest and worse than it forever after. In the robust reading
\citep{garlappi2007}, Proposition~\ref{prop:shrink}'s sign flip says the fair
ambiguity premium is negligible at $p=6$ and is the entire game at $p=420$.
Shrinkage is not a numerical hack; it is the price of not knowing $D$, and on
these calibrations the price scales with $\sqrt{p/M}$.

\begin{remark}[Relation to the virtue of complexity]
\label{rem:voc}
Read on its own, this section looks like the reverse of
\citet{kelly2024virtue}, who prove in a high-dimensional
ridge-regularized return-prediction model that expected out-of-sample
performance can keep improving as the parameter count passes the number of
observations, and that with sufficient shrinkage the non-monotonicity at the
interpolation point disappears entirely. The two pictures are the same
picture. The catastrophic column of Table~\ref{tab:est20} is the
\emph{unregularized} policy evaluated in the neighbourhood of the
interpolation point, which is exactly where the ridgeless limit of a
high-dimensional linear problem degenerates \citep{hastie2022}; and their virtue statement is conditional on shrinkage, as is
ours --- the $\sqrt{p/M}$-ridged level-$\le 2$ policy never loses to the
static one at any sample size tested. Two differences are worth keeping.
Their complexity is a free parameter, dialled by adding random features to a
fixed information set, whereas $p_{m}(d)$ here is fixed by the truncation
level and the cross-sectional width, so complexity is inherited rather than
chosen. And their asymptotics hold the complexity ratio fixed, while
Proposition~\ref{prop:model} says that in this setting the ratio can be
sidestepped: naming the driver class buys at $M/p=1.19$ what shrinkage alone
does not buy until far later. Where complexity is a virtue it is a virtue
bought with regularization --- statistical in their setting, algebraic in
ours.
\end{remark}

\section{Self-exciting allocation}
\label{sec:hawkes}

For exponential Hawkes drivers \citep{hawkes1971,bacry2015} the truncated
expected signature admits a finite-dimensional linear closure after
state-weight augmentation \citep{noguer2026paths}. Writing $Z_{t}$ for the
augmented state, the conditional objects $\mu_{t}=\mu(Z_{t})$ and
$D_{t}=D(Z_{t})$ are available in closed form, so:

\begin{proposition}[Myopic feedback policy; \textsc{Proved}]
\label{prop:feedback}
Under the closure, the mean--variance-optimal signature-linear control
re-solved at each $t$ is $\ell^{*}_{t}=(\gamma D_{t})^{-1}\mu_{t}$, a
deterministic function of the augmented state. Path-dependent allocation under
self-excitation is a linear solve in $Z_{t}$, not a stochastic control problem.
\end{proposition}

\begin{remark}[Beyond myopia; \textsc{Conjectural}]
For a general utility the candidate optimal control has the Merton form
$\ell_{t}=-\frac{V_{w}}{V_{ww}}\,D(Z_{t})^{-1}\mu(Z_{t})$, with
Proposition~\ref{prop:feedback} the constant-relative-tolerance special
case in which the hedging demand vanishes. We do not assert this: the
augmented state jumps, so the value function solves a partial
integro-differential equation rather than a diffusion HJB, and a verification
theorem on the state space of the closure is open. The myopic statement is
what is proved; the Merton form is what a companion paper should prove.
\end{remark}

This is a \emph{myopic} feedback rule --- optimal one period at a time, with no
hedging demand for future variation in $Z$. We do not solve the intertemporal
problem, and Section~\ref{sec:check} and Section~\ref{sec:cost} are both static
optimizations over a path-dependent payoff basis. Saying otherwise would
overstate what is here.

A separate warning concerns how the closure is used in practice. For a scalar
clock with baseline $\mu$, decay $\beta$ and branching ratio $n$ the
\emph{stationary} limits are $\E[N_{T}]\to\mu T/(1-n)$ and
$\Var(N_{T})\to\mu T/(1-n)^{3}$, and it is tempting to substitute them.

\begin{proposition}[Stationary substitution error; \textsc{Proved} for the
first moment, \textsc{measured} for the second]
\label{prop:critical}
Integrating the closure forward from $\lambda_{0}=\mu$ gives the exact
finite-horizon first moment
\begin{equation}
\E[N_{T}] \;=\; \lambda_{\infty}T
\;+\;(\mu-\lambda_{\infty})\,
\frac{1-e^{-\beta(1-n)T}}{\beta(1-n)},
\qquad \lambda_{\infty}=\frac{\mu}{1-n},
\label{eq:odeclosure}
\end{equation}
whose relaxation time is $\tau=(\beta(1-n))^{-1}$. The stationary limit is the
$T/\tau\to\infty$ truncation of \eqref{eq:odeclosure}; all of the error in
Table~\ref{tab:hawkes} is the truncation, not the closure.
\end{proposition}

Formula~\eqref{eq:odeclosure} and its second-moment companion are classical:
$(N_{t},\lambda_{t})$ is an affine point process whose moments solve linear
ODEs \citep{errais2010,dassios2011}, and the closure of
\citet{noguer2026paths} recovers exactly this system as the level-one block of
the state-weight-augmented expected signature. Table~\ref{tab:hawkes} separates
the two failure modes. Formula~\eqref{eq:odeclosure} reproduces the simulated
mean to within $1.3\%$ across the whole range of branching ratios, and
integrating the second-moment system likewise reproduces the simulated variance
to within $2.5\%$; the stationary substitution is wrong by $65.8\%$ in the mean
and $388.2\%$ in the variance at $n=0.9$. The closure
of \citet{noguer2026paths} is therefore not the source of the problem, and the
apparent near-critical breakdown of expected-signature methods is an artifact
of how they are usually applied. The correct instruction is not to distrust
the closure but to integrate it.

\begin{table}[ht]
\centering
\begin{tabular}{lrrrrrr}
\toprule
& \multicolumn{3}{c}{$\E[N_T]$} & \multicolumn{3}{c}{$\Var(N_T)$}\\
\cmidrule(lr){2-4}\cmidrule(lr){5-7}
$n$ & stationary & closure & simulated
    & stationary & closure & simulated\\
\midrule
$0.3$ & $14.29$  & $13.98$ & $13.99$ & $29.15$    & $27.37$   & $28.07$\\
$0.5$ & $20.00$  & $19.00$ & $19.01$ & $80.00$    & $69.00$   & $69.78$\\
$0.7$ & $33.33$  & $29.45$ & $29.35$ & $370.37$   & $257.03$  & $252.01$\\
$0.9$ & $100.00$ & $61.09$ & $60.30$ & $10000.00$ & $2080.83$ & $2048.47$\\
\bottomrule
\end{tabular}
\caption{Exponential Hawkes, $\mu=0.5$, $\beta=1$, $T=20$, $4000$ paths per
row. ``Closure'' integrates the moment system forward; ``stationary'' takes
its $T/\tau\to\infty$ limit. At $n=0.9$, $\tau=10$ and the horizon spans two
relaxation times, where the stationary variance overstates the realized one by
$4.9\times$ and the integrated closure is within $1.6\%$.}
\label{tab:hawkes}
\end{table}

\begin{figure}[ht]
\centering
\input{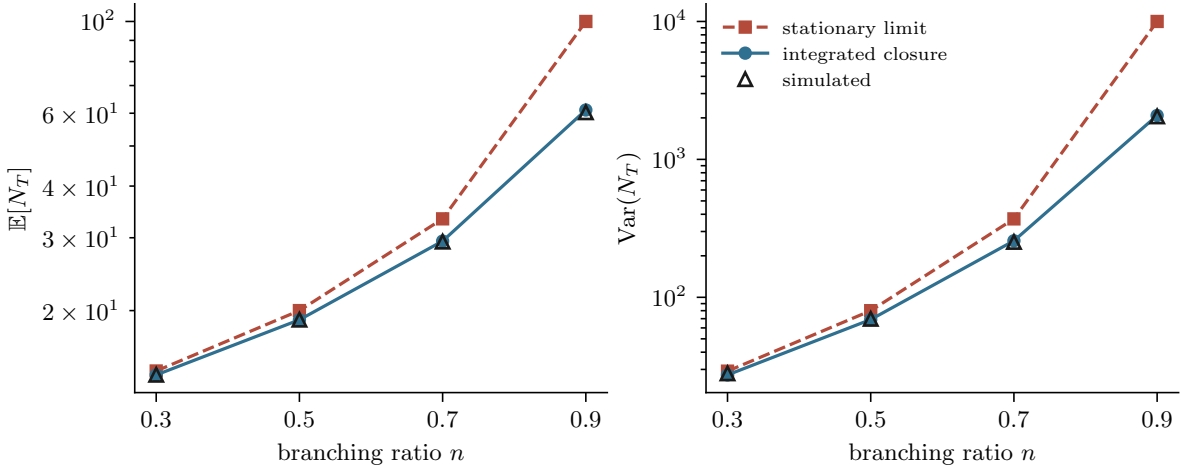}
\caption{Table~\ref{tab:hawkes} drawn: the first two moments of $N_{T}$
against the branching ratio, on a logarithmic scale. The stationary limit and
the integrated closure agree while the horizon spans many relaxation times and
part company as criticality approaches; the integrated closure tracks the
simulation across the whole range, and at $n=0.9$ the stationary variance is
almost five times the realized one. The failure is the substitution, not the
closure.}
\label{fig:hawkes}
\end{figure}

Markets are routinely measured with branching ratios in the $0.7$--$0.99$
range \citep{hardiman2013}, where the substitution overstates risk inputs
badly. The rule is blunt: substitute the stationary limit only when
$T\beta(1-\hat n)\gtrsim10$, and otherwise integrate the augmented state.

\section{Cross-area as a portfolio}
\label{sec:area}

Using the second-level cross-terms of the signature to detect lead--lag is by
now a standard idea, going back to \citet{gyurko2013} and surveyed in
\citet{lyonsmcleod2022}. The antisymmetric coordinate
$A^{ij}=\tfrac12(\Sig^{ij}-\Sig^{ji})$ is, however, not merely a statistic. It
is the payoff of the signature-linear control supported on $\{(ij),(ji)\}$ with
weights $\pm\tfrac12$ --- the lead--lag portfolio holding $X^{i}$ units of $j$
and shorting $X^{j}$ units of $i$. Its mean is
$\ip{\ell_{A}}{\E[\Sig]}$ and its variance is $\ell_{A}^{\top}D\ell_{A}$, both
from the same defect form, so its Sharpe ratio is computable in the framework
of Section~\ref{sec:defect} with no extra machinery. By
Corollary~\ref{cor:area} it is invariant to the execution convention; by the
antipode reversal property of \citet{noguer2026paths} it is reversal-odd, so
its expectation vanishes for any time-reversible driver. In
Section~\ref{sec:check} that is exactly why $\ell^{*}$ put zero weight on it.

\begin{proposition}[Cross-area detects excitation direction;
\textsc{Proved for the simulated family}]
\label{prop:area}
For a bivariate exponential Hawkes process with off-diagonal excitation
$\alpha_{ji}$ from channel $i$ into $j$, the sign of $\E[A^{12}]$ tracks the
sign of $\alpha_{21}-\alpha_{12}$, and $\E[A^{12}]$ is statistically
indistinguishable from zero under symmetric excitation.
\end{proposition}

With $\mu=(0.4,0.4)$, $\beta=(1,1)$, $T=20$, $\alpha_{\text{off}}=0.6$ and
$3000$ paths: excitation $1\to2$ gives mean cross-area $+3.253$ (s.e.\ $0.218$,
$t=14.94$); $2\to1$ gives $-3.493$ (s.e.\ $0.218$, $t=-16.02$); no cross
excitation gives $0.086$ (s.e.\ $0.183$, $t=0.47$). One scalar, no kernel
estimation.

\subsection{The counting path is not the price path}
\label{sec:signexc}

Proposition~\ref{prop:area} is a statement about \emph{activity}. The paths
there are counting processes, so the cross-area measures which channel tends to
fire first, and it is not a profit. Replacing the counting path by a price
path --- each event moving its asset by a fixed tick, with a sign --- separates
two things that the activity experiment conflates.

\begin{proposition}[Intensity excitation alone pays nothing;
\textsc{Proved for the simulated family}]
\label{prop:signexc}
Let an event in channel $i$ raise the intensity of channel $j$, and
independently let the triggered event copy the sign of its trigger with
probability $q$. The expected cross-area P\&L of the price path is null for
$q=\tfrac12$ regardless of the excitation asymmetry, and is signed by the
excitation direction only for $q\ne\tfrac12$.
\end{proposition}

\begin{table}[ht]
\centering
\begin{tabular}{lrrrr}
\toprule
configuration & mean P\&L & s.e. & $t$ & ann.\ SR\\
\midrule
no excitation, \; $q=1.00$        & $-0.000009$ & $0.000009$ & $-0.99$ & $-0.08$\\
intensity only, $1\to2$, $q=0.00$ & $+0.000011$ & $0.000012$ & $+0.96$ & $+0.08$\\
sign-carrying, $1\to2$, $q=0.85$  & $+0.000189$ & $0.000012$ & $+15.99$ & $+1.27$\\
sign-carrying, $2\to1$, $q=0.85$  & $-0.000227$ & $0.000012$ & $-18.92$ & $-1.50$\\
\bottomrule
\end{tabular}
\caption{Cross-area as traded P\&L. Bivariate Hawkes, $\mu=(0.4,0.4)$,
$T=20$, tick $0.01$, $2000$ paths per row. P\&L is in the price units induced
by the tick, so its level is not comparable across tick sizes; annualized
Sharpe ratios multiply the per-path ratio by $\sqrt{252/T}$, reading the $20$
time units as trading days.}
\label{tab:signexc}
\end{table}

The second row is the important one. Excitation that is strong enough to
produce a cross-area $t$-statistic of $15$ on the counting path produces
$t=0.96$ on the price path once the triggered moves carry no directional
information. \emph{Reflexivity in activity and reflexivity in direction are
different resources, and only the second one is tradable.} Since the
branching-ratio literature measures the first \citep{hardiman2013}, a high
measured $n$ is not by itself evidence of a lead--lag edge, and the
entropy-production ceiling of Proposition~\ref{prop:ceiling} binds against the
sign channel rather than against activity --- in the increment form of
Remark~\ref{rem:ceilscope}.

\begin{proposition}[Entropy-production ceiling; \textsc{Conditional}]
\label{prop:ceiling}
Let the observable dynamics be a stationary Markov chain with transition kernel
$p$ and stationary law $\pi$, let
$J_{ab}=\pi_{a}p_{ab}-\pi_{b}p_{ba}$ be the probability current, and let
$\sigma=\sum_{ab}\pi_{a}p_{ab}\log(\pi_{a}p_{ab}/\pi_{b}p_{ba})$ be the
per-step entropy production. For any one-step-measurable payoff $h$ with
antisymmetric part $h_{\mathrm{anti}}$, the expected directional edge is
$\tfrac12\sum_{ab}J_{ab}h_{\mathrm{anti},ab}
\le\|h_{\mathrm{anti}}\|_{\infty}\sqrt{\sigma/2}$, and it vanishes under
detailed balance.
\end{proposition}

\begin{proof}
The symmetric part of $h$ pairs to zero with the antisymmetric current, giving
the edge as $\tfrac12\langle J,h_{\mathrm{anti}}\rangle
\le\tfrac12\|h_{\mathrm{anti}}\|_{\infty}\|J\|_{1}$. The current is the
difference between the one-step law of the chain and of its time reversal, so
Pinsker's inequality \citep{coverthomas2006} gives
$\|J\|_{1}\le\sqrt{2\sigma}$. Under detailed balance $J\equiv0$.
\end{proof}

\begin{remark}[Scope of the ceiling]
\label{rem:ceilscope}
Proposition~\ref{prop:ceiling} bounds a \emph{one-step} edge, and a payoff
that is a sum of $N$ one-step-measurable increments of the same stationary
chain inherits $N\|h_{\mathrm{anti}}\|_{\infty}\sqrt{\sigma/2}$ by summation.
The cross-area is of that form only on an augmented state: writing
$A^{ij}=\sum_{k}\tfrac12(X^{i}_{k}\Delta X^{j}_{k}-X^{j}_{k}\Delta X^{i}_{k})$,
each summand is one-step measurable for the chain that carries the current
levels alongside the internal state, and that chain is stationary only under
a further assumption on the levels. So the ceiling binds the sign channel of
the increments the cross-area accumulates; it is not a bound on
$\E[A^{ij}]$ itself, and we do not claim one.
\end{remark}

The Markov premise is an assumption about the coarse-graining, not a measured
fact, and Pinsker is loose. The structural claim survives the constant:
excitation asymmetry is simultaneously the source of the edge and the entropy
production that caps it, and there is no configuration in which one is large
and the other small.

\section{Dimension and what the lift buys}
\label{sec:dim}

A control of order $m$ has $N_{m}=\sum_{k\le m}d^{k}$ coordinates; under a
geometric lift the shuffle relations cut the free directions to the free
nilpotent dimension $L_{m}=\sum_{k\le m}\frac1k\sum_{e\mid k}\mu(e)d^{k/e}$
\citep{reutenauer1993}, with $L_{m}/N_{m}$ equal to $0.500$, $0.355$, $0.263$
at $m=2,3,4$ for $d=5$ and $0.500$, $0.347$, $0.257$ for $d=10$.

\begin{proposition}[Level-two complement; \textsc{Proved}]
\label{prop:complement}
At $m=2$ the forward lift breaks the shuffle relation
$\Sig^{ij}+\Sig^{ji}=\Sig^{i}\Sig^{j}$ by exactly $-[X^{i},X^{j}]_{T}$, so the
$N_{2}-L_{2}=d(d+1)/2$ additional directions available under the forward lift
are spanned by the quadratic-covariation payoffs --- the variance and
covariance swaps. Both identities verified to $2.2\times10^{-15}$ at $d=4$,
where $N_{2}=20$, $L_{2}=10$.
\end{proposition}

\begin{remark}[\textsc{Conjectural}]
That the analogous statement holds at every $m$ --- the full $N_{m}-L_{m}$
complement being spanned by iterated-bracket functionals of the quadratic
covariation --- is plausible from the Hopf square of \citet{noguer2026paths}
but is not proved here, and nothing in this paper depends on it.
\end{remark}

The practical content is that choosing a lift is choosing whether the
volatility surface is in the investable universe. A geometric-lift theory is
more parsimonious and cannot express variance trading; a forward-lift theory is
larger and can, at the cost of $d(d+1)/2$ more parameters --- which,
by Section~\ref{sec:cost}, is not a cost to be waved through.

\section{The program in five identifications}

\begin{enumerate}
\item \textbf{Risk} $=$ the defect form, read off $\E[\Sig^{\le2m}]$
(Proposition~\ref{prop:defect}).
\item \textbf{Rebalancing premium} $=$ lift gap $=$ excess growth
(Corollary~\ref{cor:excess}); match the lift to how you execute.
\item \textbf{Gap risk} exists only under the forward lift
(Corollary~\ref{cor:ruin}).
\item \textbf{Direction} $=$ cross-area: convention-free
(Corollary~\ref{cor:area}), signed by excitation asymmetry
(Proposition~\ref{prop:area}), capped by entropy production
(Proposition~\ref{prop:ceiling}).
\item \textbf{Ambition} $=$ what the effective dimension allows
(Proposition~\ref{prop:floor}), with the regularization dictated by the
sample ratio rather than by preference (Proposition~\ref{prop:shrink}) --- and
with the floor itself a property of unstructured estimation rather than of the
words (Proposition~\ref{prop:model}). What the words buy, on a zero-area
driver, is convexity and not path-dependence
(Proposition~\ref{prop:wordclass}).
\end{enumerate}

\paragraph{Separation of obstructions.} The three difficulties of the title
enter through three identities that do not interact at the orders computed:
geometry through $D=\Cov$ and its Hilbert closure
(Proposition~\ref{prop:defect}, Theorem~\ref{thm:fredholm}); execution
through $\tfrac12[X,X]_{T}$ contracted with
$\operatorname{diag}(\pi)-\pi\pi^{\top}$
(Theorem~\ref{thm:lift}, Corollary~\ref{cor:excess}); estimation through the
exact loss $\tfrac{\gamma}{2}\|D^{1/2}(\hat\ell-\ell^{*})\|^{2}$ at rate
$\operatorname{tr}(D)/(\lambda_{0}M)$ (Proposition~\ref{prop:celoss},
Theorem~\ref{thm:consistency}). We
state this as a synthesis rather than a theorem: each identity is exact on its
own, and their joint content is that each obstruction has its own currency ---
a tensor block, a portfolio tensor, a sample ratio --- and none can be paid in
another's.

\section{Claim tiering and non-claims}

\paragraph{Proved.} Proposition~\ref{prop:defect} (known) and
Corollary~\ref{cor:psd}; Theorem~\ref{thm:fredholm},
Corollary~\ref{cor:frontier}, Propositions~\ref{prop:logdefect}
and~\ref{prop:celoss}; Corollary~\ref{cor:distance},
Proposition~\ref{prop:kelly} to second order,
Proposition~\ref{prop:robust}, and Theorem~\ref{thm:consistency} under its
hypotheses (i)--(ii); Theorem~\ref{thm:lift} and
Corollaries~\ref{cor:excess}--\ref{cor:ruin};
Proposition~\ref{prop:feedback} as a myopic statement;
Proposition~\ref{prop:complement} at $m=2$; Proposition~\ref{prop:critical}
for the first moment, with its second-moment analogue integrated numerically
rather than in closed form; the dimension counts.

\paragraph{Conditional.} Propositions~\ref{prop:floor},
\ref{prop:shrink}, \ref{prop:wordclass} and~\ref{prop:model} are measured on
two calibrations at one risk aversion; the thresholds are indicative, not
universal. Proposition~\ref{prop:ceiling} rests on the Markov and
one-step-measurability premises stated with it, is false without them, and by
Remark~\ref{rem:ceilscope} applies to the increments of the cross-area rather
than to the cross-area itself. Propositions~\ref{prop:area}
and~\ref{prop:signexc} are proved for the simulated families only. Both the
Marchenko--Pastur localization of the estimation barrier and the appeal to
\citet{kanzhou2007} are analogies: the first has an exact rank statement
attached, the second is exact only for the level-one sub-problem, since
level-two payoffs are quadratic in Gaussians and $\hat D$ is not Wishart. The
effective-dimension reading of Remark~\ref{rem:effdim} is a measured
regularity across two calibrations, not a theorem. The Merton form of the
remark in Section~\ref{sec:hawkes} is conjectural.

\paragraph{Conjectural.} The general-$m$ complement
(Remark after Proposition~\ref{prop:complement}); and that the $M/p$ threshold
is stable across drift and covariance structures other than the two used here.

\paragraph{Non-claims.} No market data anywhere; all parameters invented. The
Gaussian driver in Sections~\ref{sec:check} and~\ref{sec:cost} is i.i.d.\
across replications, so the estimation results measure sampling error only and
understate the difficulty under non-stationarity. Transaction costs appear
nowhere; the level-two words are exactly the high-turnover ones, so
Section~\ref{sec:cost} is an upper bound on realizable value. Neither ridge
intensity was tuned; a tuned or nonlinear shrinkage \citep{ledoit2020} would
dominate both. The
model-consistent policy of Proposition~\ref{prop:model} is evaluated under a
correctly specified driver and therefore measures the cost of not using a
model, never the cost of using a wrong one; nothing here bounds its bias under
misspecification. The word-class attribution of
Proposition~\ref{prop:wordclass} is for drivers with zero expected area, which
is why the antisymmetric block earns nothing in it; a driver with excitation
asymmetry moves value into that block, and Section~\ref{sec:area} measures
that separately rather than jointly. The cross-area magnitudes of Proposition~\ref{prop:area} are activity, not
P\&L; Section~\ref{sec:signexc} gives the traded version, where the tick size
is fixed and the sign-copying probability $q$ is imposed rather than estimated
from any market. Nothing shows any strategy attains the entropy bound.

\paragraph{Scope.}
The theoretical conclusions should be interpreted as a framework for organizing
path-dependent portfolio construction rather than as a claim of immediate
economic profitability. The simulations isolate estimation error and structural
effects in controlled environments. A market implementation would additionally
require transaction-cost modeling, turnover constraints, robust estimation,
out-of-sample validation, and comparison against simpler benchmark strategies.

\bibliographystyle{plainnat}
\bibliography{pathport}

\begin{thebibliography}{43}
\providecommand{\natexlab}[1]{#1}
\providecommand{\url}[1]{\texttt{#1}}
\expandafter\ifx\csname urlstyle\endcsname\relax
  \providecommand{\doi}[1]{doi: #1}\else
  \providecommand{\doi}{doi: \begingroup \urlstyle{rm}\Url}\fi

\bibitem[Bacry et~al.(2015)Bacry, Mastromatteo, and Muzy]{bacry2015}
Emmanuel Bacry, Iacopo Mastromatteo, and Jean-Fran\c{c}ois Muzy.
\newblock Hawkes processes in finance.
\newblock \emph{Market Microstructure and Liquidity}, 1\penalty0 (1), 2015.

\bibitem[Bun et~al.(2017)Bun, Bouchaud, and Potters]{bun2017}
Jo\"el Bun, Jean-Philippe Bouchaud, and Marc Potters.
\newblock Cleaning large correlation matrices: Tools from random matrix theory.
\newblock \emph{Physics Reports}, 666:\penalty0 1--109, 2017.

\bibitem[Chen(1957)]{chen1957}
Kuo-Tsai Chen.
\newblock Integration of paths, geometric invariants and a generalized
  {B}aker--{H}ausdorff formula.
\newblock \emph{Annals of Mathematics}, 65\penalty0 (1):\penalty0 163--178,
  1957.

\bibitem[Chevyrev and Lyons(2016)]{chevyrev2016}
Ilya Chevyrev and Terry Lyons.
\newblock Characteristic functions of measures on geometric rough paths.
\newblock \emph{Annals of Probability}, 44\penalty0 (6):\penalty0 4049--4082,
  2016.

\bibitem[Cover(1991)]{cover1991}
Thomas~M. Cover.
\newblock Universal portfolios.
\newblock \emph{Mathematical Finance}, 1\penalty0 (1):\penalty0 1--29, 1991.

\bibitem[Cover and Thomas(2006)]{coverthomas2006}
Thomas~M. Cover and Joy~A. Thomas.
\newblock \emph{Elements of Information Theory}.
\newblock Wiley, 2nd edition, 2006.

\bibitem[Cuchiero and M\"oller(2023)]{cuchiero2023moller}
Christa Cuchiero and Janka M\"oller.
\newblock Signature methods in stochastic portfolio theory.
\newblock \emph{arXiv:2310.02322}, 2023.

\bibitem[Cuchiero et~al.(2023)Cuchiero, Gazzani, and
  Svaluto-Ferro]{cuchiero2023signature}
Christa Cuchiero, Guido Gazzani, and Sara Svaluto-Ferro.
\newblock Signature-based models: Theory and calibration.
\newblock \emph{SIAM Journal on Financial Mathematics}, 14\penalty0
  (3):\penalty0 910--957, 2023.

\bibitem[Dassios and Zhao(2011)]{dassios2011}
Angelos Dassios and Hongbiao Zhao.
\newblock A dynamic contagion process.
\newblock \emph{Advances in Applied Probability}, 43\penalty0 (3):\penalty0
  814--846, 2011.

\bibitem[DeMiguel et~al.(2009)DeMiguel, Garlappi, and Uppal]{demiguel2009}
Victor DeMiguel, Lorenzo Garlappi, and Raman Uppal.
\newblock Optimal versus naive diversification: How inefficient is the $1/n$
  portfolio strategy?
\newblock \emph{Review of Financial Studies}, 22\penalty0 (5):\penalty0
  1915--1953, 2009.

\bibitem[Engl et~al.(1996)Engl, Hanke, and Neubauer]{engl1996}
Heinz~W. Engl, Martin Hanke, and Andreas Neubauer.
\newblock \emph{Regularization of Inverse Problems}.
\newblock Kluwer, 1996.

\bibitem[Errais et~al.(2010)Errais, Giesecke, and Goldberg]{errais2010}
Eymen Errais, Kay Giesecke, and Lisa~R. Goldberg.
\newblock Affine point processes and portfolio credit risk.
\newblock \emph{SIAM Journal on Financial Mathematics}, 1\penalty0
  (1):\penalty0 642--665, 2010.

\bibitem[Fernholz(2002)]{fernholz2002}
E.~Robert Fernholz.
\newblock \emph{Stochastic Portfolio Theory}.
\newblock Springer, 2002.

\bibitem[Fernholz and Karatzas(2009)]{fernholzkaratzas2009}
E.~Robert Fernholz and Ioannis Karatzas.
\newblock Stochastic portfolio theory: An overview.
\newblock In \emph{Handbook of Numerical Analysis: Mathematical Modeling and
  Numerical Methods in Finance}, pages 89--167. Elsevier, 2009.

\bibitem[Futter et~al.(2025{\natexlab{a}})Futter, Horvath, and
  Wiese]{futter2023sig}
Owen Futter, Blanka Horvath, and Magnus Wiese.
\newblock Signature trading: A path-dependent extension of the mean--variance
  framework with exogenous signals.
\newblock \emph{Quantitative Finance}, 25\penalty0 (2):\penalty0 197--226,
  2025{\natexlab{a}}.
\newblock arXiv:2308.15135.

\bibitem[Futter et~al.(2025{\natexlab{b}})Futter, {Muca Cirone}, and
  Horvath]{mucacirone2025}
Owen Futter, Nicola {Muca Cirone}, and Blanka Horvath.
\newblock Kernel learning for mean--variance trading strategies.
\newblock \emph{arXiv:2507.10701}, 2025{\natexlab{b}}.

\bibitem[Garlappi et~al.(2007)Garlappi, Uppal, and Wang]{garlappi2007}
Lorenzo Garlappi, Raman Uppal, and Tan Wang.
\newblock Portfolio selection with parameter and model uncertainty: A
  multi-prior approach.
\newblock \emph{Review of Financial Studies}, 20\penalty0 (1):\penalty0 41--81,
  2007.

\bibitem[Gyurk\'o et~al.(2013)Gyurk\'o, Lyons, Kontkowski, and
  Field]{gyurko2013}
Lajos~Gergely Gyurk\'o, Terry Lyons, Mark Kontkowski, and Jonathan Field.
\newblock Extracting information from the signature of a financial data stream.
\newblock \emph{arXiv:1307.7244}, 2013.

\bibitem[Hardiman et~al.(2013)Hardiman, Bercot, and Bouchaud]{hardiman2013}
Stephen~J. Hardiman, Nicolas Bercot, and Jean-Philippe Bouchaud.
\newblock Critical reflexivity in financial markets: A {H}awkes process
  analysis.
\newblock \emph{European Physical Journal B}, 86:\penalty0 442, 2013.

\bibitem[Hastie et~al.(2022)Hastie, Montanari, Rosset, and
  Tibshirani]{hastie2022}
Trevor Hastie, Andrea Montanari, Saharon Rosset, and Ryan~J. Tibshirani.
\newblock Surprises in high-dimensional ridgeless least squares interpolation.
\newblock \emph{Annals of Statistics}, 50\penalty0 (2):\penalty0 949--986,
  2022.

\bibitem[Hawkes(1971)]{hawkes1971}
Alan~G. Hawkes.
\newblock Spectra of some self-exciting and mutually exciting point processes.
\newblock \emph{Biometrika}, 58\penalty0 (1):\penalty0 83--90, 1971.

\bibitem[Hoffman(2000)]{hoffman2000}
Michael~E. Hoffman.
\newblock Quasi-shuffle products.
\newblock \emph{Journal of Algebraic Combinatorics}, 11\penalty0 (1):\penalty0
  49--68, 2000.

\bibitem[Kalsi et~al.(2020)Kalsi, Lyons, and {Perez Arribas}]{kalsi2020}
Jasdeep Kalsi, Terry Lyons, and Imanol {Perez Arribas}.
\newblock Optimal execution with rough path signatures.
\newblock \emph{SIAM Journal on Financial Mathematics}, 11\penalty0
  (2):\penalty0 470--493, 2020.

\bibitem[Kan and Zhou(2007)]{kanzhou2007}
Raymond Kan and Guofu Zhou.
\newblock Optimal portfolio choice with parameter uncertainty.
\newblock \emph{Journal of Financial and Quantitative Analysis}, 42\penalty0
  (3):\penalty0 621--656, 2007.

\bibitem[Kelly et~al.(2024)Kelly, Malamud, and Zhou]{kelly2024virtue}
Bryan~T. Kelly, Semyon Malamud, and Kangying Zhou.
\newblock The virtue of complexity in return prediction.
\newblock \emph{Journal of Finance}, 79\penalty0 (1):\penalty0 459--503, 2024.

\bibitem[Kir\'aly and Oberhauser(2019)]{kiraly2019}
Franz~J. Kir\'aly and Harald Oberhauser.
\newblock Kernels for sequentially ordered data.
\newblock \emph{Journal of Machine Learning Research}, 20\penalty0
  (31):\penalty0 1--45, 2019.

\bibitem[Ledoit and Wolf(2004)]{ledoit2004}
Olivier Ledoit and Michael Wolf.
\newblock A well-conditioned estimator for large-dimensional covariance
  matrices.
\newblock \emph{Journal of Multivariate Analysis}, 88\penalty0 (2):\penalty0
  365--411, 2004.

\bibitem[Ledoit and Wolf(2020)]{ledoit2020}
Olivier Ledoit and Michael Wolf.
\newblock Analytical nonlinear shrinkage of large-dimensional covariance
  matrices.
\newblock \emph{Annals of Statistics}, 48\penalty0 (5):\penalty0 3043--3065,
  2020.

\bibitem[Levin et~al.(2013)Levin, Lyons, and Ni]{levin2013}
Daniel Levin, Terry Lyons, and Hao Ni.
\newblock Learning from the past, predicting the statistics for the future,
  learning an evolving system.
\newblock \emph{arXiv:1309.0260}, 2013.

\bibitem[Lyons and McLeod(2022)]{lyonsmcleod2022}
Terry Lyons and Andrew~D. McLeod.
\newblock Signature methods in machine learning.
\newblock \emph{arXiv:2206.14674}, 2022.

\bibitem[Lyons and Ni(2015)]{lyonsni2015}
Terry Lyons and Hao Ni.
\newblock Expected signature of {B}rownian motion up to the first exit time
  from a bounded domain.
\newblock \emph{Annals of Probability}, 43\penalty0 (5):\penalty0 2729--2762,
  2015.

\bibitem[Lyons et~al.(2020)Lyons, Nejad, and {Perez
  Arribas}]{lyons2019nonparametric}
Terry Lyons, Sina Nejad, and Imanol {Perez Arribas}.
\newblock Nonparametric pricing and hedging of exotic derivatives.
\newblock \emph{Applied Mathematical Finance}, 27\penalty0 (6):\penalty0
  457--494, 2020.

\bibitem[Lyons and Qian(2002)]{lyonsqian2002}
Terry~J. Lyons and Zhongmin Qian.
\newblock \emph{System Control and Rough Paths}.
\newblock Oxford University Press, 2002.

\bibitem[Marcus(1981)]{marcus1981}
Steven~I. Marcus.
\newblock Modeling and approximation of stochastic differential equations
  driven by semimartingales.
\newblock \emph{Stochastics}, 4\penalty0 (3):\penalty0 223--245, 1981.

\bibitem[Markowitz(1952)]{markowitz1952}
Harry Markowitz.
\newblock Portfolio selection.
\newblock \emph{Journal of Finance}, 7\penalty0 (1):\penalty0 77--91, 1952.

\bibitem[Mar\v{c}enko and Pastur(1967)]{marchenko1967}
Vladimir~A. Mar\v{c}enko and Leonid~A. Pastur.
\newblock Distribution of eigenvalues for some sets of random matrices.
\newblock \emph{Mathematics of the USSR-Sbornik}, 1\penalty0 (4):\penalty0
  457--483, 1967.

\bibitem[Michaud(1989)]{michaud1989}
Richard~O. Michaud.
\newblock The {M}arkowitz optimization enigma: Is optimized optimal?
\newblock \emph{Financial Analysts Journal}, 45\penalty0 (1):\penalty0 31--42,
  1989.

\bibitem[Muandet et~al.(2017)Muandet, Fukumizu, Sriperumbudur, and
  Sch\"olkopf]{muandet2017}
Krikamol Muandet, Kenji Fukumizu, Bharath Sriperumbudur, and Bernhard
  Sch\"olkopf.
\newblock Kernel mean embedding of distributions: A review and beyond.
\newblock \emph{Foundations and Trends in Machine Learning}, 10\penalty0
  (1--2):\penalty0 1--141, 2017.

\bibitem[{Noguer i Alonso}(2026)]{noguer2026paths}
Miquel {Noguer i Alonso}.
\newblock A general theory of paths: Signatures, jump lifts, and expected
  signatures of self-exciting processes.
\newblock \emph{arXiv:2606.28869 [math.PR]}, 2026.

\bibitem[Reutenauer(1993)]{reutenauer1993}
Christophe Reutenauer.
\newblock Free {L}ie algebras.
\newblock \emph{London Mathematical Society Monographs}, 7, 1993.

\bibitem[Salvi et~al.(2021)Salvi, Cass, Foster, Lyons, and Yang]{salvi2021}
Cristopher Salvi, Thomas Cass, James Foster, Terry Lyons, and Weixin Yang.
\newblock The signature kernel is the solution of a {G}oursat {PDE}.
\newblock \emph{SIAM Journal on Mathematics of Data Science}, 3\penalty0
  (3):\penalty0 873--899, 2021.

\bibitem[Tropp(2012)]{tropp2012}
Joel~A. Tropp.
\newblock User-friendly tail bounds for sums of random matrices.
\newblock \emph{Foundations of Computational Mathematics}, 12\penalty0
  (4):\penalty0 389--434, 2012.

\bibitem[Tropp(2015)]{tropp2015}
Joel~A. Tropp.
\newblock \emph{An Introduction to Matrix Concentration Inequalities}, volume~8
  of \emph{Foundations and Trends in Machine Learning}.
\newblock Now Publishers, 2015.

\end{thebibliography}

\end{document}